\documentclass[10pt, conference, compsocconf]{IEEEtran}

\usepackage[cmex10]{amsmath}

\IEEEoverridecommandlockouts
\usepackage{mdwmath}
\usepackage{mdwtab}
\begin{document}
%\title{Bare Demo of IEEEtran.cls for IEEECS Conferences}

\author{\IEEEauthorblockN{Yaron Singer \thanks{This work was done while the author was at UC Berkeley and supported by a Microsoft Research Fellowship.  This paper appeared in the Proceedings of the 51st Annual IEEE Symposium on Foundations of Computer Science (FOCS 2010).}}
\IEEEauthorblockA{School of Engineering and Applied Sciences\\
Harvard University\\
 yaron@seas.harvard.edu}}

\title{Budget Feasible Mechanisms}

\setlength{\textwidth}{6.5in} \setlength{\textheight}{9in}
\setlength{\oddsidemargin}{0in} \setlength{\evensidemargin}{0in}
\setlength{\hoffset}{0in} \setlength{\voffset}{0in}
\setlength{\marginparsep}{0in} \setlength{\marginparwidth}{0in}
\setlength{\topmargin}{0in} \setlength{\headheight}{0in}
\setlength{\headsep}{0in}
\renewcommand{\baselinestretch}{.96}

\bibliographystyle{plain}

\newtheorem{theorem}{Theorem}[section]
\newtheorem{lemma}[theorem]{Lemma}
\newtheorem{proposition}[theorem]{Proposition}
\newtheorem{definition}[theorem]{Definition}

\date{}
\maketitle

\begin{abstract}
We study a novel class of mechanism design problems in which the
outcomes are constrained {\em by the payments}.  This basic class
of mechanism design problems captures many common economic situations,
and yet it has not been
studied, to our knowledge, in the past.  We focus on the case of
procurement auctions in which sellers have private costs, and the
auctioneer aims to maximize a utility function on subsets of items,
under the constraint that the sum of the payments provided by the
mechanism does not exceed a given budget.  Standard mechanism design
ideas such as the VCG mechanism and its variants are not applicable
here.  We show that, for general functions, the budget constraint
can render mechanisms arbitrarily bad in terms of the utility of the
buyer.  However, our main result shows that for the important class of 
submodular functions, a bounded approximation ratio is achievable.  Better
approximation results are obtained for subclasses of the submodular functions.  
We explore the space of budget feasible mechanisms in other domains and give a characterization 
under more restricted conditions.
\end{abstract}

\section{Introduction}

Consider the following familiar problems:

\begin{itemize}
\item{\bf Knapsack:} Given a budget $B$ and a set of items $\mathcal{N}=\{1,\ldots,n\}$, each with cost $c_i$ and value $v_i$, find a subset of items $S$ which maximizes $\sum_{i\in S}v_i$ under the budget constraint.
\item{\bf Matching:} Given a budget $B$ and a bipartite graph, with set of edges $\mathcal{N}=\{e_1,\ldots,e_n\}$ each with cost $c_e$ and value $v_e$, find a legal matching $S$ which maximizes  $\sum_{e\in S}v_e$ under the budget constraint.
\item{\bf Coverage:} Given a budget $B$ and subsets $\mathcal{N} = \{T_1,\ldots,T_n\}$ of some ground set, each with a cost $c_i$ find a subset $S$ which maximizes $|\cup_{i \in S} T_i |$ under the budget constraint.

\end{itemize}

Three much studied, and much solved, optimization problems. However,
suppose that the elements of $\mathcal{N}$ are not combinatorial
objects, but {\em strategic agents} with private \emph{costs}.
The above problems then capture natural economic interactions:
Knapsack, for example, models a simple procurement auction, while
Coverage may model the problem of maximizing exposure of an advertising
campaign under a budget.  These are precisely the kinds of
economic interactions we wish to study here: reverse auctions with
private costs, with the goal of optimizing the auctioneer's value.

At first glance it may seem that the problem we describe falls
within the scope of familiar domains.
However, closer inspection reveals a new dimension of
difficulty: the budget constraint applies not to the \emph{costs}
but to the \emph{payments the mechanism uses to support
truthfulness}.  We need mechanisms whose sum of payments never
exceeds the given budget. \newline

\begin{itemize}
\item[]{\em Can we design mechanisms that implement these
intended economic interactions in the most favorable way to the
auctioneer without their payments exceeding the budget?}\newline
\end{itemize}

Mechanism design is by now a very mature discipline  
and the recent injection of
computational thought has helped develop it even further, and in new
and forward-looking directions~\cite{NR01}.  Procurement auctions, 
introduced to computer scientists already in~\cite{NR01}, were at first 
studied under {\em utilitarian} objectives, seeking to optimize social welfare~\cite{NR01,FPSS02}.
More recently procurement auctions have been studied under the non-utilitarian framework of {\em frugality}~\cite{CFHK08,AT07,KKT05, ESS04,T03}
--- essentially, payment optimization in reverse auctions.  

These situations still fall within classical mechanism
design theory, where {\em the set of possible outcomes is a priori
fixed and publicly known.}  By ``set of all possible outcomes'' here
we mean the set of all possible allocations, with payments projected
out.  In other words, there is a rich class of allocations,
independent of payments, each of which is realizable by a truthful
mechanism.  In the three introductory examples, however, the set of possible
outcomes is {\em not} fixed {\em or} publicly known: It depends
crucially on the participants' private information, ultimately on
the mechanism's payments.  It is this peculiarity that makes these
three problems difficult, and places them at a blind spot of
mechanism design.  

\subsection*{Budget Feasible Mechanisms}
We say a truthful mechanism is \emph{budget feasible} if its payments do not exceed a given budget.  
In single parameter domains, where each agent's private information is a single number, designing truthful 
mechanisms often reduces to designing monotone allocation rules, since payments can be computed via binary 
search~\cite{MN08}.  This no longer holds when the payments are restricted by a budget:
Designing a budget feasible allocation rule requires understanding its payments, which in-turn depend 
on the allocation rule itself.  Not surprisingly, it seems that budget feasible mechanisms are very tricky to find.

\paragraph*{The VCG mechanism does not work}
Consider a simple Knapsack instance where all items have identical values,
and except for one item whose cost equals the budget, all items have small costs.
The VCG mechanism will choose the $n-1$ small-cost agents, paying the budget to each. 
Thus, while this mechanism returns the optimal solution with total {\em
cost} within the budget, the total \emph{payment} will be way over
budget (in fact, $(n-1)$ times the budget).

\paragraph*{In general, nothing can work}  
Consider a slight variation of the above problem, in which all items have small costs, and identical values \emph{as long as a particular item 
$i$ is in the solution}, and otherwise all have value 0 (for example, think of $i$ as a corkscrew and the rest of the items as bottles of wine).  How well can a budget feasible mechanism do here?  If the mechanism has a bounded approximation ratio
it must always guarantee to include $i$ in its solution.  This however implies that as long as $i$ declares a cost that is less than the mechanism's budget, the mechanism includes her in the solution.  A truthful mechanism must therefore surrender its entire budget to $i$.  This of course results in an unbounded approximation ratio.

\subsection*{Our Results}

The question, then, is:  Which classes of functions 
have budget feasible mechanisms with good approximation properties?
{\em Our main result is a randomized constant factor budget feasible mechanism
that is universally truthful for the quite general, and important, class of nondecreasing 
submodular functions} (Theorem~\ref{main}).  

For a slightly broader class, that of fractionally subadditive functions, 
we show that computational constraints dictate a lower bound.  
As shown in the simple example above, superadditive functions bring out the clash between truthfulness and the budget constraint.  On a positive note, the three problems in the beginning of the section correspond to 
subclasses (additive, OXS, and coverage) of submodular maximization problems.
We show improved approximations for these problems and other special cases.  
We further explore the space of budget feasible mechanisms, showing several impossibilities as well as a characterization under more restricted conditions.

\subsection*{Related Work}
\paragraph*{Budgets in auctions} Budgets came under scrutiny in auction theory
\cite{DLN08,BLM08,FMPS07,BCIMS05} after observing behavior of
bidders in online automated auctions \cite{Google08}, as well as in
spectrum auctions where bidding is performed by groups of strategic
experts \cite{BLM08}.  While these pioneering works highlight the
significance and challenges that budgets introduce to mechanism
design, they relate to an entirely different concept than the one we
study here.  While these works study the impact of budgets on
strategic \emph{bidders}, our interest is to explore the budget's
effect on the \emph{mechanism}.  These papers, however, do point out
the complexity induced by budget constraints in mechanism design, and the
need for approximations.

\paragraph*{Frugality} In recent years a theory of {\em frugality} has
been developed with the goal of providing mechanisms for procurement auctions
that admit minimal payments~\cite{CFHK08,AT07,KKT05,ESS04,T03} .  Budget feasibility and frugality are complementary concepts.
Frugality is about buying a feasible solution at minimum cost --- there are
no preferences among the solutions, and the goal is to minimize
payments.  In our setting we have no preferences among payments ---
as long as they are below the budget --- but we do care about
the value of the solutions.  The two approaches are complementary
also in another important sense:  in our last section we show that
for all the problems studied in the frugality literature there are
no budget feasible mechanisms.

\paragraph*{Cost Sharing} Somewhat conceptually closer to our work is the
subject of cost sharing, in which agents have private values for a
service, there is a nondecreasing cost for allocating the service to
agents, and the goal is to maximize the agents' valuations under the
cost (see~\cite{JM07} for a survey). The proportional share mechanism 
we study in this paper is inspired by~\cite{MS01} and~\cite{RS06}.
The relationship between cost sharing and our setting, however, is quite limited: we are not aiming 
to optimize a function of the agents private information under a public cost function, 
but rather optimize a public function under constraints dictated by agents' private information and a fixed budget.  
Here, our goal is non-utilitarian --- we aim to maximize the \emph{buyer's} demand,
which is independent of the agents' utilities.

\paragraph*{Submodular Maximization} From a pure algorithmic perspective, even under a cardinality constraint, maximizing a submodular function is well known to be NP-hard, and an $1-1/e$ approximation ratio can be achieved by greedily taking items based on their marginal contribution~\cite{FNW78}.   When items have costs, variations of greedy on marginal contribution normalized by cost can achieve constant factor approximations, and even the optimal $1-1/e$ ratio \cite{KMN99,KG05}.  For submodular maximization problems that can be expressed as integer programs, rounding solutions of linear and nonlinear programs can, in some cases, achieve the optimal constant approximation ratio~\cite{AS04}.  

\subsection*{Paper Organization}
After the necessary definitions in Section~\ref{sec:model}, we
present a mechanism for the class of symmetric submodular functions
(Section~\ref{sec:sym}); this special case simplifies the problem enormously
and facilitates the introduction of ideas and intuition for the
general submodular case. Our main result for submodular functions is developed
in Section~\ref{sec:SM}.  Finally, in Section~\ref{sec:exp}, we further discuss the space 
of budget feasibility, improved approximations, impossibility results and characterization.

\section{The Model}\label{sec:model}

In a {\em budget-limited reverse auction} we have a
set of items $[n]=\{1,\ldots,n\}$, and a single buyer.  Each
item $i \in [n]$ is associated with a cost $c_{i}\in \mathcal{R}_{+}$,
while the buyer has a budget $B \in \mathcal{R}_{+}$  and a demand valuation function
$V:2^{[n]} \to \mathcal{R_+}$.  In the {\em full information} case, costs are common knowledge,
and the objective is to maximize the demand function under the budget, i.e. find the subset $S \in \{T | \sum_{i \in T}c_{i} \leq B\}$
for which $V(S)$ is maximized.

We focus on the {\em strategic case}, in which each item is held by a unique agent and
costs are \emph{private}.  The budget and demand function of the buyer are common knowledge.
A solution is a subset and a payment vector, and the objective is to maximize the demand function 
while the \emph{payments} (not costs) are within the budget.  
More formally, a mechanism $\mathcal{M}=(f,p)$ consists of an allocation function $f:\mathcal{R}_+^{n} \to 2^{[n]}$
and a payment function $p:\mathcal{R}_+^{n} \to \mathcal{R}_+^{n}$.
The allocation function $f$ maps a set of $n$ bids to a subset
$S=f(c_1,\ldots,c_n)\subseteq [n]$.  The payment function $p$
returns a vector $p_{1},\ldots,p_{n}$ of payments to the agents. We
shall often omit the arguments $c_1,\ldots,c_n$ when writing $f$ and
$p$.  We shall denote by $s_{1},\ldots,s_{n}$ the characteristic
vector of $S$, that is, $s_{i} = 1$ iff $i \in S$.  As usual, we seek normalized ($s_i=0$ implies $p_i=0$), individually rational ($p_{i} \geq s_i\cdot c_{i}$) mechanisms with no positive transfers ($p_i\geq 0$).  As it is common in algorithmic mechanism design, our goal is
manifold. We seek mechanisms that are:

\begin{enumerate}
\item {\bf Truthful,}  that is, reporting the true costs is a dominant strategy for sellers.
Formally, a mechanism $\mathcal{M}=(f,p)$ is {\em truthful} ({\em incentive
compatible}) if for every $i \in [n]$ with cost $c_{i}$ and bid $c'_{i}$, and every set of bids by $[n]\setminus\{i\}$
we have $p_{i} - s_{i}\cdot
c_{i} \geq p'_{i} - s'_{i}\cdot c_{i}$, where $(s_{i},p_{i})$ and
$(s'_{i},p'_{i})$ are the allocations and payments when the bidding
is $c_{i}$ and $c'_{i}$, respectively.  A mechanism that is a randomization over truthful mechanisms is \emph{universally truthful}.

\item {\bf Computationally Efficient.} The functions $f$ and $p$ can be computed
in polynomial time.  In cases where the demand function requires exponential data to be represented (as in the general submodular case), we take the common ``black-box'' approach and assume the buyer has access to an oracle which allows evaluating any subset $S \subseteq [n]$, with polynomially many queries.  Such queries are known as value queries. This is a weaker model than ones allowing demand or general queries (see~\cite{BN05} for a definition).  Since our main interest here is algorithmic, this strengthens our results.

\item {\bf Budget Feasible.}  Importantly, we require that a
mechanism's allocation rule and payments do not
exceed the budget: $\sum_{i}p_{i}s_{i} \leq B$.  We call such mechanisms \emph{budget feasible}.

\item {\bf Approximation.}  We want the allocated subset to yield the highest possible value for the buyer.
For $\alpha \geq 1$ we say that a mechanism is  $\alpha$-approximate if the mechanism allocates to a
set $S$ such that $OPT(c,\mathcal{N},B)\leq \alpha V(S)$, where $OPT(c,\mathcal{N}, B)$ denotes the value of full information optimal solution over a set of agents $\mathcal{N}$ with cost vector $c$ and budget $B$.  As usual, when dealing with randomization we seek mechanisms that yield constant factor approximations in expectation. 

\end{enumerate}

This is a {\em single parameter} mechanism design problem, in that each bidder has only one private
value.  We shall repeatedly rely on Myerson's
well-known characterization
\footnote{Note that although there is a budget constraint on the payments, Myerson's characterization applies to our setting as well.  Due to the characterization, we know that the allocation rule determines the payment function.  The budget constraint can therefore be viewed as a property of the allocation rule alone.}:\newline

\begin{theorem}[\cite{M81}]
In single parameter domains a normalized mechanism $\mathcal{M}=(f,p)$ is
truthful iff:
\begin{itemize}
\item[(i)] {\bf $f$ is monotone:}  $\forall i \in[n]$, if $c'_i\leq c_i$ then $i\in f(c_i,c_{-i})$ implies $i\in f(c'_i,c_{-i})$ for every $c_{-i}$;
\item[(ii)] {\bf winners are paid threshold payments:} payment to each winning bidder is $\inf{\{c_i: i\notin f(c_i,c_{-i})\}}.$
\end{itemize}
\end{theorem}~

\section{Symmetric Submodular Functions}\label{sec:sym}

We now introduce a subclass of submodular functions which is devoid of many of the intricacies of the general case.
It will serve as an exposition of the basic ideas, and will help understand the difficulties in the
general case.    

We say a set function is {\em symmetric} if it only depends on the
cardinality of the set, rather than the identity of the
items it holds. Symmetric submodular functions (also called
\emph{downward sloping}) were used by Vickrey 
in his seminal work on multi-unit auctions~\cite{V61}.  They have a very simple structure:\newline

\begin{definition}
A function $V:2^{[n]}\to \mathcal{R_{+}}$ is symmetric submodular if there exist $r_1\geq
\ldots \geq r_n\geq 0$, such that $V(S) = \sum_{i =1}^{|S|} r_i$.
\end{definition}~

Consider the following allocation rule $f_{\cal M}$:  Sort the $n$ bids so that $c_1\leq c_2\leq\ldots\leq c_n$, and consider
the largest $k$ such that $c_k\leq {B / k}$.  That is, $k$ is the place where the curve of the increasing costs 
intersects the hyperbola $B / k$.  The set allocated here is $\{1,2,\ldots,k\}$.  That is, 
$f_{\cal M}=\{1,2,\ldots, k\}$.  This is obviously a monotone allocation rule:  an agent cannot be excluded when 
decreasing her bid.  In the Appendix we show that paying each agent $\theta_i = \min\{B/k, c_{k+1}\}$ results in a truthful mechanism.\footnote{It is rather interesting that the second term is needed; we show in the Appendix that the mechanism breaks down in its absence.}

Observe that this allocation rule has the property we seek: summing over the payments that support truthfulness satisfies the budget constraint.    
Hence this gives us a budget feasible mechanism. 
Importantly, this is also a good approximation of the optimum solution:\newline

\begin{theorem}\label{thm:ds}
The above mechanism has approximation ratio of two.
\end{theorem}~

\begin{proof}
Observe that the optimal solution is obtained
by greedily choosing the lowest-priced items until the
budget is exhausted.  By the downward sloping property, to
prove the result  it suffices to show that the mechanism returns at least 
half of the items in the greedy solution.  Assume for purpose of contradiction that the optimum solution has $\ell$ items,
and the mechanism returns less than $\ell / 2$.
It follows that  $c_{\lceil \ell/2 \rceil}
> {2B}/{\ell} $.   Note however, that this is impossible since we
assume that $c_{\lceil \ell/2 \rceil} \leq \ldots \leq c_{\ell}$, and
$\sum_{i=\lceil \ell/2 \rceil}^{\ell}c_{i} \leq B$ which implies that
$c_{\lceil \ell/2 \rceil} \leq {2B}/{\ell}$, a contradiction.
\end{proof}~

In Section~\ref{sec:exp} we show that no better approximation ratio is possible.
This is rather surprising, given the simplicity of the full-information problem,
and illustrates the intricacies of budget feasibility.

\section{General Submodular Functions}\label{sec:SM}
We now turn to the general case of nondecreasing submodular functions.  A demand function $V$ is nondecreasing if $S \subseteq T$ implies $V(S) \leq V(T)$.\newline

\begin{definition}
$V:2^{[n]}\to R_{+}$ is submodular if $V(S \cup \{i\}) - V(S) \geq V(T \cup \{i\}) - V(T) \qquad \forall S \subseteq T.$
\end{definition}~

In general, submodular functions may require exponential data to be represented.  We therefore assume 
the buyer has access to a \emph{value oracle} which given a query $S\subseteq [n]$ returns $V(S)$ (see related work section for more discussion on submodular maximization).  In designing truthful mechanisms for submodular maximization problems, the greedy approach is a natural fit, since it is monotone when agents are sorted according to their increasing marginal contributions relative to cost: the marginal contribution of an agent $i$ given a subset $S$ is $V_{i |S}:= V(S \cup \{i\}) -V(S)$.  In the marginal contribution-per-cost sorting the $i+1$ agent is the agent $j$ for which  $V_{j | S_i }/ c_{j}$ is maximized over all agents $\mathcal{N}$ where $S_i = \{1,2\ldots,i\}$, and $S_{0}=\emptyset$.  To simplify notation we will write $V_{i}$ instead of $V_{i| S_{i-1}}$.  This sorting, in the presence of submodularity, implies:
\begin{eqnarray}
V_1 / c_{1} \geq V_2 / c_{2} \geq \ldots \geq  V_n / c_{n}.\label{sorting}
\end{eqnarray}
Notice that $V(S_{k}) = \sum_{i \leq k}V_{i}$ for all $k$.

\subsection{The Proportional Share Allocation Rule}

The mechanism from the previous section for the limited symmetric case can be generalized appropriately to work for various classes in the submodular family of functions.\newline

\begin{definition}
For a budget $B$ and set of agents $\mathcal{N}$ with cost vector $c$, the generalized proportional share allocation rule, denoted $f_{\mathcal{M}}(c,B,\mathcal{N})$ sorts agents according to (\ref{sorting}) with costs vector $c$ and budget $B$ and allocates to agents $\{1,\ldots,k\}$ that respect $c_{i} \leq B\cdot V_{i}/ V(S_{i})$.  
Observe that this condition is met for every $\{1,\ldots,i\}$ when $i \leq k$.  
\end{definition}~

For concreteness consider the case of additive valuations (Knapsack from the Introduction): each agent is associated with a fixed value $v_i$ and $V(S)=\sum_{i\in S}v_i$.  Here the marginal contribution of each agent is independent of their place in the sorting, and we simply have that $V_i=v_i$ for all agents $i \in[n]$.  In this case $f_{\mathcal{M}}$ produces a budget-feasible mechanism. The reason is, it assures us that for each agent $i$, the threshold payments of $f_{\mathcal{M}}$, denoted $\theta_{i}$ do not  exceed the agent's proportional share:

$$\theta'_i = \min \Big\{ \frac{V_i\cdot B} {\sum_{i\in S}V_i}, \frac{V_i   \cdot c_{k+1}}{V_{k+1}} \Big \}.$$ 

\noindent which allows budget feasibility, as well as individually rationality: $\theta'_i\leq c_i$.  This seems to make the proportional share allocation rule an ideal candidate to obtain budget feasible mechanisms.  Indeed, with some minor adjustments,  for many problems with functions in the submodular class (e.g. symmetric, Knapsack, Matching,) this general approach works well and produces budget feasible mechanisms with good approximation guarantees (see the following section for more details).  Furthermore, as we discussed above, the proportional share mechanism is optimal is some cases, and in some restricted environments our characterizations show that this is essentially the \emph{only} budget feasible mechanism (see the following section).     
It seems however that this natural approach completely fails as soon as we encounter more involved cases, such as Coverage.

\subsection{The Difficulties}
The Coverage problem captures many of the difficulties that are associated with designing budget feasible mechanisms for the general submodular case.  In Coverage the marginal contribution of an agent is not fixed, but depends on the subset allocated by the algorithm in the previous stages.  An agent's marginal contribution therefore depends on its position in the sorting, which introduces several difficulties.

\paragraph*{Marginal Contributions are Affected by Costs}  
When applying the proportional share mechanism in Coverage, paying agents  $\theta'_i$ as above will be under the budget, as we desire.  However, observe that for each agent $i$ the payment depends on the marginal contribution $V_i$, which is determined by $i$'s position in the sorting.  Thus, in such a case the payments will depend on the agent's declared cost, and therefore cannot induce truthfulness, making the proportional share mechanism hopeless here.  

Simple allocation and payment schemes that are independent of the agent's position in the sorting also fail.  An approach that may seem natural is to replace marginal contributions with Shapley values~\cite{JM07} since they make the proportional contribution of an allocated agent independent of her position in the sorting.  Unfortunately, such an approach cannot approximate better than a factor of $\sqrt{n}$, as we show in the full version of this paper.  While it is tempting to get rid of the marginal contribution sorting, it is the only known means for obtaining good approximation guarantees for the general submodular case.

\paragraph*{Non-monotonicity of the Maximum Operator}  Bounded approximation ratios for submodular maximization under a knapsack constraint depend crucially on the ability to take the maximum between a greedy solution and the item with highest value.  In the general case, as well in the case of Coverage, taking this maximum does not preserve monotonicity:  simple examples show that for allocation rules that depend on marginal contribution sorting, by decreasing her cost an agent can \emph{decrease} the value of the allocation.

\subsection{Overview of Our Approach}
Our approach is based in three ideas:

\begin{itemize} 
\item  First, we derive an alternative characterization of the threshold payments of the proportional share allocation rule.  Since we know that this rule does not work, this may seem futile.  We'll show that this characterization plays a significant role in our design.      

\item  Using the above characterization, we show that for any (nondecreasing) submodular function, the threshold payments of a slightly modified version of the proportional share allocation rule are ``not too far'' from the agents' proportional contributions.  This enables us to guarantee that when running the modified version of the proportional share allocation rule \emph{with a constant fraction of the budget}, the threshold payments will be budget feasible.  

\item Finally, to obtain the approximation guarantee we partition the agents in a manner that allows us to include the variation of the proportional share rule over a \emph{subset of agents} and obtain good approximation guarantees. 
\end{itemize}

\subsection{Characterizing Threshold Payments}\label{sec:payments}

The following definition is key in our characterization.\newline

\begin{definition}
The marginal contribution of agent $i$ at point $j$ is $V_{i(j)} := V(T_{j-1} \cup \{i\}) - V(T_{j-1})$ where $T_{j}$ denotes the subset of the first $j$ agents in the marginal-contribution-per-cost sorting (as in (\ref{sorting})) over the subset $\mathcal{N} \setminus \{i\}$.
\end{definition}~

The intuition behind the payments characterization can be described as follows.  Consider running the proportional share mechanism without agent $i$.  For the first $j$ agents in the marginal contribution sorting, using the marginal contribution of $i$ at point $j$ we can find the maximal cost that agent $i$ can declare in order to be allocated instead of the agent in the $j$th place in the sorting.  While these costs may have arbitrary behavior as a function of $j$, we will show that taking the maximum of these values guarantees payments that support truthfulness.  To avoid confusion we use $T_{j}$ to denote the first $j$ agents according to this sorting, $V'_{j}$ to denote the marginal contribution of the $j$th agent in this case, and $k'$ to denote the index of the last agent $j\in \mathcal{N}\setminus\{i\}$ that respects $c_{j} \leq V'_{j} \cdot B/V(T_{j})$.  For brevity we will write $c_{i(j)} := V_{i(j)}\cdot c_{j}/V'_{j}$ and $\rho_{i(j)} := V_{i(j)}\cdot B/V(T_{j-1} \cup \{i\})$.\newline

\begin{lemma}[Payments Characterization]\label{lem:char}
The threshold payment for $f_{\mathcal{M}}$ is
$$\theta_{i} = \max_{j \in [k'+1]} \Big \{\min\{c_{i(j)},\rho_{i(j)}\} \Big\}.$$
\end{lemma}~

\begin{proof}
To characterize the threshold payment for agent $i$, relabel the agents according to the marginal-contribution-per-cost sorting over the subset $\mathcal{N} \setminus \{i\}$.  Consider $f_{\mathcal{M}}$ as a sequential allocation rule: at each stage $j$, the mechanism resorts the remaining agents that have not yet been allocated according to marginal contribution-per-cost sorting, and allocates to the first agent in the sorting if she meets the condition $c_{j} \leq V'_{j}\cdot B/V(T_{j-1} \cup \{j\})$.  

For a given stage $j$ in this sequential allocation, we can find the maximal cost agent $i$ would have been able declare and be allocated, if she had been considered by the mechanism at this stage: The value $c_{i(j)} = V_{i(j)} \cdot c_{j} / V'_{j}$ is the maximal cost $i$ can declare which would place her ahead of $j$ in the sorting, and if this cost does not exceed $\rho_{i(j)} = B\cdot V'_{i(j)}/V(T_{j-1} \cup \{i\})$, the mechanism would have allocated to agent $i$.  Therefore, had $i$ appeared at stage $j$, the minimum between these values is the maximal cost she can declare and be allocated at this stage.  Since $V_{i(j)}$ monotonically decreases with $j$ while $c_{j} / V'_{j}$ increases, $c_{i(j)}$ may have arbitrary behavior as a function of $j$.  However, as we now show, taking the maximum of these values result in threshold payments. 

Let $r$ be the index in $[k'+1]$ for which $\min\{{c}_{i(j)},\rho_{i(j)}\}$ is maximal.  Declaring a cost below $\theta_{i} \leq {c}_{i(r)}$ guarantees $i$ to be within the first $r \leq k'+1$ elements in the sorting stage of the mechanism, with $r-1$ items allocated.  Since $\theta_{i} \leq \rho_{i(r)}$, $i$ will be allocated. 

To see that declaring a higher cost prevents $i$ from being allocated, consider first the case where $c_{i(r)}  \leq \rho_{i(r)}$.  A higher cost places $i$ after $r$ in the sorting stage of the mechanism.  If the maximum of $c_{i(j)}$ over all $j\in [k'+1]$ is $c_{i(r)}$, reporting a higher cost places $i$ after an element which is not allocated and therefore it will not be allocated.  Otherwise, if $c_{i(r)} < c_{i(j)}$, for some $j \leq k'+1$, by the maximality of $r$ it must be the case that: 

$$ \frac{B \cdot V_{i(j)} }  {V(T_{j-1} \cup \{i\})  } = \rho_{i(j)}  <    c_{i(r)} <  c_{i(j)}$$ 

\noindent and $i$ will not be allocated as a cost above $\rho_{i(j)}$ will not meet the allocation condition.

In the second case when $c_{i(r)}  > \rho_{i(r)}$, if $r$ is the index which maximizes $\rho_{i(j)}$ over all indices in $[k'+1]$, reporting a higher cost will not meet the mechanism's allocation condition at each index in $[k'+1]$.  Otherwise, if there is some other index $j\in [k'+1]$ for which this maximum is achieved, then: 

$$ \frac{V_{i(j)}\cdot c_{j}}{V'_{j}}= c_{i(j)}  < \rho_{i(r)} < \rho_{i(j)}$$ 

\noindent and thus declaring a higher cost in this case places $i$ after $j$ in the sorting, and the mechanism will not consider $i$.
\end{proof}~

\begin{lemma}[Individual Rationality]\label{ir}
The mechanism $f_{\mathcal{M}}$ is individually rational, i.e., $c_i \leq \theta_{i}$.
\end{lemma}~

\begin{proof}
Observe that:
\begin{itemize}
\item[(a)] $V_{i(j)} \geq V_{i(j+1)} \quad \forall j \in \mathcal{N}$;
\item[(b)] $T_{j} = S_{j}                   \quad \forall j < i$;
\item[(c)] $V_{i | T_{i-1} } = V_{i}$.
\end{itemize}

Since the threshold payment is the maximum over all $\min\{c_{i(j)},\rho_{i(j)}\}$ in $[k'+1]$, it is enough to show that $c_i \leq \min\{{c}_{i(j)},\rho_{i(j)}\}$ for a certain $j \leq k'+1$.  Since (b) implies that $i \leq k'+1$ we can consider $i$'s replacement $j$ which appears in the $i$th place in the marginal-contribution-per-cost sorting over $\mathcal{N}\setminus\{i\}$.  Since $i \in [k]$, and due to (b) and (c) above, we have that

$$c_{i} \leq \frac{V_i\cdot B}{V((S_{i-1} \cup \{i\})}= \frac{V_{i | T_{i-1}}\cdot B }{V(T_{i-1}\cup \{i\})} = c_{i(j)}.$$

In the original sorting, $i$ appears ahead of $j$ (as implied from (b)), and therefore its relative marginal contribution is greater.  Thus:

$$c_{i} \leq \frac{V_{i | S_{i-1} }\cdot c_{j}}{V_{j | S_{i-1}}} = \frac{V_{i | T_{i-1}}\cdot c_{j}}{V'_{j | T_{j-1}}} = \rho_{i(j)}.$$

It therefore follows that $c_{i} \leq \min \{c_{i(j)},\rho_{i(j)}\} \leq \theta_{i}.$
\end{proof}~

\subsection{Payment Bounds}
The characterization above allows us to include a slightly modified version of the proportional share allocation rule in our mechanism, with threshold payments that are guaranteed to be no more than a constant factor away from agents' proportional contribution.  This is a key property which guides the design of our mechanism.

We will run the modified proportional share allocation rule over a subset of the agents $\mathcal{N}_{s}$, with a constant fraction of the budget, denoted $B'$.  We describe $\mathcal{N}_{s}$ and $B'$ explicitly in the following section, but for the purpose of showing the payment bounds, we can think of these as any subset of agents and any budget.  In this modified version of the proportional share allocation rule, for $i_{s}^* := \textrm{argmax}_{j \in \mathcal{N}_{s}}V(\{j\})$, we sort the agents of $\mathcal{N}_{s} \setminus\{i_{s}^*\}$ according to the marginal-contribution-per-cost order, and allocate to $W = S_{k} \cup \{i_{s}^*\}$, where $S_{k}$ are all $k$ agents in $\mathcal{N}_{s}$ that respect the condition $c_{i} \leq V_{i}\cdot B/ V(S_{i} \cup \{i_{s}^*\})$.  The characterization from above easily extends to this case using $\rho_{i(j)} = V_{i(j)}\cdot B/V(T_{j-1} \cup \{i,i^*_{s}\})$.  Under this modification we can show the following desirable bound on the threshold payments:\newline

\begin{lemma}[Payment Bounds]\label{bound}
For $i \in W \setminus\{i_{s}^*\}$: $$\theta_{i} \leq  \Big(\frac{6e-2}{e-1}\Big) \frac{{V_{i}\cdot B'}}{{V(W)}}.$$
\end{lemma}~

\begin{proof}
For $T_{k'}$ as above, let $W'=T_{k'} \cup \{i^{*}_{s}\}$ and let $r$ be the index for which $\theta_{i} = \min\{c_{i(r)}, \rho_{i(r)} \}$.  If $r \leq k'$, observe that the sorting implies $c_{r}/V'_{r} \leq c_{k'}/V'_{k'}$ and therefore: 

$$\theta_{i} \leq  \frac{V_{i(r)}\cdot c_{r}}{V'_{r}} \leq \frac{V_{i(r)}\cdot c_{k'}}{V'_{k'}} \leq \frac{V_{i(r)}\cdot B'}{V(W')} \leq \frac{V_{i}\cdot B'}{V(W')}$$

\noindent where the last inequality relies on the observation that $V_{i(j)} \leq V_{i}$ for every $j\in \mathcal{N}_{s} \setminus\{i^*_{s}\}$, which is due to the fact that $T_{j}=S_{j}$ for $j \leq i$ and the decreasing marginal utility property of $V$.  In the case where if $r=k'+1$: 
$$\theta_{i} \leq \rho_{i(r)} = \frac{V_{i(k'+1)} \cdot B'  }{V(W' \cup \{i\})} \leq \frac{V_{i} \cdot B'}{V(W')}.$$

\noindent We therefore see that in both cases $\theta_{i} \leq {V_{i} \cdot B'}/{V(W')}$.  To complete our proof, we will show that $V(W) \leq ((5e-1)/(e-1))V(W')$.  

Consider the marginal contribution-per-cost sorting on $\mathcal{N}_{s}\setminus\{i,i^*_{s}\}$ and let $\ell' := \textrm{max} \{t |\sum_{j \leq t}c_{j} \leq B'\}$.  Due to the marginal contribution-per-cost sorting we have:

$$ \frac{c_{k'+1}}{V'_{k'+1}} \sum_{j=k'+1}^{\ell'} V'_{j} \leq \sum_{j=k'+1}^{\ell'} \Big ( \frac{c_{j}}{V'_{j}} \Big )V'_{j} \leq B'   .$$

Since $c_{k'+1}/V'_{k'+1} > B / V(W' \cup \{k'+1\}) $, the above inequality implies: 

$$V(T_{\ell'}) - V(T_{k'}) < V(W' \cup \{k'+1\})$$

\noindent and since $T_{k'} \subset W'$, and $i_{s}^* \in W'$ this implies that $3V(W') \geq V(T_{\ell})$.  From submodularity, it can be shown that $V(T_{\ell+1})$ is a ${e}/({e-1})$-approximation of the optimal solution (over $\mathcal{N}_{s}\setminus\{i,i^*_{s}\}$)~\cite{KMN99,KG05}.  Since $i_{s}^* \in W'$ we have that:
\begin{align}
V(W) \ & \leq  \ OPT(c,B',\mathcal{N}_{s}) 			\nonumber \\ 
	   & \leq \ OPT(B',\mathcal{N}_{s} \setminus\{i,i^*\}) + V(\{i\})+V(\{i^*\})  \nonumber \\
	  & < \ \Big (\frac{4e}{e-1} + 2\Big)V(W')		 \nonumber
\end{align}

\noindent which concludes our proof.
\end{proof}~

\subsection{Approximation Guarantee}
To obtain the approximation guarantee we partition the set of agents based on their value and declared cost.  Let $i^*:= \textrm{argmax}_{i\in \mathcal{N}}V(\{i\})$, $\mathcal{N}_{s} :=\{i \neq i^*: c_{i}\leq B/2\}$, and $\mathcal{N}_{\ell} := \mathcal{N} \setminus (\mathcal{N}_{s} \cup \{i^*\})$.  This partition, as we argue later, does not break monotonicity.  Our mechanism will run the modified proportional share allocation rule over $\mathcal{N}_{s}$, with a constant fraction of the budget $B/\alpha$.  

The rational behind this approach is the following.  Due to their large cost, any feasible solution includes at most one agent in $\mathcal{N}_{\ell}$, and since each such agent has value lower than that of $i^*$, we will be able to discard $\mathcal{N}_{\ell}$ without sacrificing too much from the quality of our solution.   Since all the costs in $\mathcal{N}_{s}$ are below $B/2$, including $i^*_{s}$ in our solution $W$ is always feasible.  Importantly, we can show this allocation rule is a constant factor approximation over $\mathcal{N}_{s}$:\newline

\begin{lemma}\label{apx_g}
$\Big( \frac{(3 +2\alpha) e -1}{e-1}\Big)  V(W) \geq OPT(c,B,\mathcal{N}_{s})$.
\end{lemma}~

\begin{proof}[Sketch] 
Similarly to what we have shown in lemma~\ref{bound}, one can show that 
$$V(S_\ell) - V(S_k) < \alpha V(W \cup \{k+1\}) $$ 
\noindent where $\ell$ is the maximal index for which $\sum_{j\leq \ell}c_{j} \leq B$ when items are taken according to their marginal contribution relative to cost on $\mathcal{N}_{s}\setminus\{i_{s}^*\}$.  Since $W$ includes $i^*_{s}$, one can verify that $V(S_{\ell+1}) \leq (2+2\alpha)V(W)$.  Similar to before, we know that 
$(1-1/e)OPT(c,B,\mathcal{N}_{s} \setminus \{i^*_{s}\}) \leq V(S_{\ell+1})$, and thus:
\begin{align}
OPT(c,B,\mathcal{N}_{s}) & \ \leq \ OPT(c,B,\mathcal{N}_{s}\setminus\{i^*_{s}\}) + V(\{i^*_{s}\}) \nonumber \\
& \ \leq \ \Big ( \frac{(2+2\alpha)e}{e-1}+1\Big)V(W) \nonumber 
\end{align}
\noindent which implies our desired bound.
\end{proof}~

\subsection{Main Result}
We can now prove our main theorem.\newline

\begin{theorem}\label{main}
For any submodular maximization problem there exists a constant factor approximation randomized mechanism in the value query model which is budget feasible and universally truthful.  Furthermore, no budget feasible mechanism can do better than $2-\epsilon$, for any fixed $\epsilon>0$.
\end{theorem}~

Our analysis shows that in expectation our mechanism guarantees an approximation ratio of  $\Big(\frac{58e^2-32e+6}{(e-1)^2}\Big)\approx117.7$.
It is possible that tighter analysis can show the mechanism does better.  We further discuss this points in Section~\ref{sec:disc}.\newline

\begin{proof}
Consider the following mechanism:\newline

\begin{tabular}{|l|}
\hline
{\bf  A Budget Feasible Approximation Mechanism} \\ \hspace{0.35in}{\bf for Submodular Functions}\\
\hline
{\bf Initialize:}
\hspace{0.15in} $i^*\longleftarrow \textrm{argmax}_{i \in \mathcal{N}}V(\{i\}),$\\ 
\hspace{0.73in} $\mathcal{N}_s \longleftarrow   \{i \neq i^*:c_{i}\leq B/2\}$,\\
\hspace{0.73in} $i^*_{s} \longleftarrow \textrm{argmax}_{i \in \mathcal{N}_s}V(\{i\})$,\\
\hspace{0.73in} $W \longleftarrow \{i^*_{s}\}$,\\
\hspace{0.73in}$B' \longleftarrow \Big((e-1)/(12e-4)\Big) B$,\\
\hspace{0.78in}$i \longleftarrow \textrm{argmax}_{j \in \mathcal{N}_s\setminus\{i^*_{s}\}} V(\{j\}) / c_{j}$\\\\

{\bf While} $c_i \leq V_{i}\cdot B'/V(W \cup \{i\}) $\\
\qquad  {\bf Do:} $W \longleftarrow W \cup \{i\}$\\
\hspace{1.6cm}  $i \longleftarrow \textrm{argmax}_{j \in \mathcal{N}\setminus\{i^*_{s}\}}V_{j|S_{i}}/ c_{j}$\\ \\
{\bf Output:} {Choose u.a.r. from } $\{(W,\hat{\theta}) ,(\{i^*\},B) \}$\\

\hline
\end{tabular}\newline

The payment $\hat{\theta}$ here is $\hat{\theta}_{i^*_{s}} = B/2$ and $\hat{\theta}_{i} = \min\{\theta_i,B/2\}$ for $i\neq i^*_{s}$, where $\theta_i$ are the payments as described for the modified 
proportional share allocation rule, using the budget $B' = B/\alpha$, for $\alpha = (12e-4)/(e-1)$.  Observe that $\alpha$ is exactly twice the constant from the bound in Lemma~\ref{bound}.  

In case $i^*$ is allocated, $B$ is clearly her threshold payment.  If $W$ is allocated, from the characterization lemma and the fact that $\mathcal{N}_{s}$ consists only of agents with cost less than $B/2$, $\hat{\theta}$ as described above are clearly the threshold payments.  One can verify the partition of agents is monotone and since payments are bounded by $B/2$ agents in $\mathcal{N}_{\ell}$ cannot benefit by misreporting their cost. 

Since the modified proportional share allocation rule uses $B'=B/\alpha$ as its budget,  from lemma~\ref{bound}, we can conclude that: 

$$\sum_{i}\hat{\theta}_{i} \leq B/2 +  \alpha \sum_{i \neq i^{*}_{s}} \frac{V_{i}\cdot B'}{V(W)} \leq B$$ 

\noindent and the mechanism is therefore truthful and budget feasible.  Individual rationality and monotonicity were discussed above, and the lower bound as described in the following section applies here as well.

Finally, let $\beta = ((3+2\alpha) / (e - 1) ) / (e-1)$.  From Lemma ~\ref{apx_g} and the partition of agents we have:

\begin{align}
OPT(c,B,\mathcal{N}) & \leq \ OPT(c,B,\mathcal{N} \setminus \mathcal{N}_{\ell}) + V(\{i^*\}) \nonumber \\
& \leq \ OPT(c,B,\mathcal{N}_{s}) + 2V(\{i^*\}) \nonumber \\
& \leq \beta V(W) + 2V(\{i^*\}) \nonumber \\
& \leq (2+ \beta) \max\Big \{V(W), V(\{i^*\})\Big\} \nonumber 
\end{align}

\noindent which gives us our desired approximation ratio.
\end{proof}~

\section{The Space of Budget Feasible Mechanisms}\label{sec:exp}
In this section we address some natural questions that arise when exploring budget feasible mechanisms.  We limit our discussion to deterministic mechanisms.
  
\subsection{Lower Bounds and Improved Approximations}
For special cases of submodular functions, specialized techniques based on the main result yield better approximation ratios: \newline

\begin{theorem}
For Knapsack there is a budget feasible 5-approximation mechanism.   For Matching there is a budget feasible $(\frac{5e-1}{e-1})$-approximation mechanism.  For either problem, no budget feasible mechanism can approximate within a factor better than $2-\epsilon$, for any fixed $\epsilon>0$.
\end{theorem}~

We defer the details of upper bound proofs to the full version of the paper.  We will now show a lower bound that is independent of computational assumptions and shows that no approximation ratio better than two is possible.  For our lower bound we'll consider a very simple function, $V(S) = |S|$.  Observe that all the functions in the submodular class we have discussed, including symmetric submodular include this demand valuation.   \newline

\begin{proposition}\label{ds_lb}
For $V(S) = |S|$, no budget feasible mechanism can guarantee an approximation of $2-\epsilon$, for any $\epsilon>0$.
\end{proposition}~

\begin{proof}
Suppose we have $n$ items with costs
$c_{1}=c_2=\cdots = c_{n} = B/2+\epsilon$, for some positive $\epsilon< B/2$.
Assume, for purpose of contradiction that $f$ is a budget feasible mechanism with approximation ratio better than 2.  In particular, $f$ has a finite approximation ratio and must therefore allocate to at least one agent in this case.  W.l.o.g., assume $f$ allocates to agent $1$.  

By monotonicity, agent $1$ can reduce her cost to $c'_{1}= \epsilon' < B/2-\epsilon$ and remain allocated.  For this cost vector, $(c_{1},c_{-1})$,  Myerson's characterization implies that the threshold payment for agent 1 should be at least $B/2 +\epsilon$, by individual rationality and budget feasibility, $f$ cannot allocate to any other agent.  Observe however that the optimal full information solution in this case allocates to two agents which contradicts $f$'s approximation ratio guarantee.
\end{proof}~

\subsection{Lower Bound on Fractionally Subadditive Functions} 
In light of the positive results for submodular functions, our natural desire would be to extend the main result to more general classes of problems. 
We now show there is little hope in that, at least in the value query model.  Let us consider fractionally subadditive functions:\newline
  
\begin{definition}
A function
$V:2^{[n]}\to \mathcal{R}$ is called fractionally subadditive if there exists a finite set of additive valuations $\{a_{1},\ldots,a_{t}\}$ s.t. $V(S) = max_{i \in [t]}a_i(S)$.
\end{definition}~

It is known that every submodular function can be represented as a fractionally subadditive function, and that all fractionally subadditive functions are subadditive~\cite{LLN01}.  Using a simple reduction from~\cite{MSV08} we show that for mechanisms which use value query oracles, obtaining reasonable approximations in the case of fractionally subadditive demands is hard, regardless of incentive considerations.\newline

\begin{theorem}
In the case of fractionally subadditive demands, any algorithm which approximates within a factor better than $n^{\frac{1}{2}-\epsilon}$, for any fixed $\epsilon>0$, requires exponentially many value queries.  This is true even in the setting where all costs are public knowledge.
\end{theorem}~

\begin{proof} Our proof relies on the following lemma shown in~\cite{MSV08}.  We will use $v_{S}$ to denote the additive function that assigns the value 1 to an item $j \in S$ and 0 to all other items, and $\overline{v}$ to denote the function that assigns a value $(1+\epsilon/2)/(n^{\frac{1}{2}})$ to every item $j \in [n]$.\newline

\begin{lemma}[\cite{MSV08}]
There exists a set $T$, $|T|=\sqrt{n}$,  for which distinguishing between the functions $V = \max\{v_{S:|S|\leq(1+\epsilon/2)n^{\epsilon}},\overline{v}\}$, and $V' = \max\{V,v_{T}\}$ requires exponentially (in $n^{\epsilon}$) many value queries.
\end{lemma}~

Let $c_{1}=c_{2},\ldots,c_{n} = 1$ and $B=\sqrt{n}$.  For such $T$ as in the lemma, we have that $\sum_{i\in T}c_{i}=B$, and $V'(T)=\sqrt{n}$ while $V(T)=(1+\epsilon/2)n^{\epsilon}$.  Thus, an algorithm that approximates better than the desired approximation ratio must be able to distinguish between these two valuations, and requires exponentially many value queries.
\end{proof}~

\subsection{Impossibility of Hiring a Team of Agents}  Looking at problems previously studied in procurement, one could also ask whether ``hiring a team of agents'' problems~\cite{AT07,T03,KKT05,CFHK08} have good budget feasible mechanisms.  In these problems there is a set of feasible outcomes (e.g. all possible spanning trees or all paths from $s$ to $t$) and the goal is to design a mechanism that yields a feasible solution.  In the literature these problems have been studied with the goal of designing mechanisms which yield minimal payments (frugal), according to various benchmarks.  

In our notation, such problems can be written as having a function $V(S)=1$ if $S \in \mathcal{F}$, and 0 otherwise, where $\mathcal{F}$ is the set of all feasible outcomes.  Call such a problem {\em nontrivial} if all solutions in $\mathcal{F}$ contain more than one element.\newline

\begin{theorem}
There is no budget feasible mechanism with a bounded approximation ratio for any nontrivial ``hiring a team of agents'' problem.
\end{theorem}~

\begin{proof} 
Assume for purpose of contradiction, there exists a budget feasible mechanism $f$ that guarantees a bounded approximation ratio.  Let $S$ be a feasible outcome, and  consider the bid profile in which all agents in $S$ declare positive cost $\epsilon<B/|S|$, and all other agents declare $B$.  Since the problem is nontrivial, the minimal cost of a feasible outcome different than $S$ (if it exists) exceeds the budget.  Since $f$ guarantees a bounded approximation ratio, it must allocate to $S$.  For agent $i \in S$, consider the cost vector $(c'_{i},c_{-i})$ with $c'_{i} = B- \epsilon(|S|-1)$.  Observe that $S$ remains a cost-feasible solution, and the only one in $\mathcal{F}$, and therefore $f(c'_{i},c_{-i}) =S$, which implies the threshold payment of agent $i$ is at least $c'_{i}$.  Since this holds for all agents in $S$, budget feasibility is contradicted.
\end{proof}~

\subsection{Characterizing Budget Feasible Mechanisms in Restricted Settings}
The characterization of budget feasible mechanisms with good approximation ratios naturally depends on the environment in which the mechanisms are implemented, or more concretely, the function we aim to optimize.  In characterizations, we often introduce additional restrictions on our mechanisms which we then use as guidelines in their design or for complete characterization.
We now consider mechanisms that respect the two additional conditions of anonymity~\cite{ADL09} and weak stability (similar to~\cite{DS08}).  Informally, a mechanism is weakly stable if an agent doesn't hurt the rest when reducing her cost, and anonymous if its allocation rule does not depend on the agents' identities.\newline  

\begin{definition}
An allocation rule $f$ satisfies anonymity if $i \in f(c_{i},c_{j},c_{-ij})$ implies $j \in f(c'_{i},c'_{j},c_{-ij})$ when $c'_{i}=c_{j},c'_{j}=c_{i}$.  
\end{definition}~

\begin{definition}
An allocation rule $f$ satisfies weak stability if for every $i,j \in f(c_{i},c_{-i})$, $c'_{i} \leq c_{i}$ implies $j \in f(c'_{i},c_{-i})$.
\end{definition}~

Note that for the class of symmetric submodular functions, the proportional share mechanism respects these conditions.  Furthermore, for symmetric submodular functions it seems quite reasonable that mechanisms with good approximation properties satisfy these conditions.  The obvious question here is whether all budget feasible mechanisms for the symmetric submodular case take the form of the proportional share allocation rule.\newline
\begin{theorem}
Let $f$ be a budget feasible mechanism that is anonymous and weakly stable, and let $S = f(c)$ for some bid profile $c$.  Then, for all $i\in S$ it must be  that $c_{i} \leq B/|S|$.
\end{theorem}~

\begin{proof}
Assume for purpose of contradiction that there is a bid profile $c=(c_{1},\ldots,c_{n})$, s.t. $f(c)=S$ and there is some $i \in S$ for which $c_i > B/|S|$.  
Let $c'$ be the bid profile in which all agents in $S \setminus \{i\}$ bid $c_{min} = \min_{j \in S}c_{j}$, and the rest bid as in $c$.  Since $f$ is weakly stable, we have that $S \subseteq f(c')$.  Let $c''$ the bid profile where $i$ bids $c_{min}$ as well, and the rest of the agents bid as in $c'$.  From monotonicity we have that $i$ is allocated, and again $S\subseteq f(c'')$.  We now claim that under the profile $c''$, the threshold price for each agent in $S$ is at least $c_{i} > B/|S|$.  To see this, observe that $i$'s threshold price must be at least $c_{i}$, since $i \in f(c')$.  Since $f$ is anonymous, and all agents in $S$ declare the same price, the threshold price for each agent in $S$ must also be at least $c_{i}$.  Thus, payments to agents in $S$ exceed the budget, contradicting budget feasibility.
\end{proof}~

\section{Discussion}\label{sec:disc}
The space of budget feasible mechanisms seems quite broad and invites for further investigation.  The richness of the submodular class implies there are many problems for which better approximation ratios are achievable.  We have briefly discussed some specific examples which include Knapsack, Matching, Coverage and the symmetric submodular case, and believe there are many more interesting problems to study.  Furthermore, we believe a better approximation ratio is achievable for the general case, and it may be possible to improve the analysis of the upper bound by showing tighter payment bounds.

While we have made a first step towards characterization by considering more restricted mechanisms, we believe a much more general characterization for budget feasible mechanisms awaits to be revealed. 
Finally, it would be interesting to further explore the lower bounds dictated by budget feasibility.  Here we've shown several lower bound which are independent of computational assumptions, and it would be interesting to extend these techniques. 

\section*{Acknowledgements}
The author wishes to thank Christos Papadimitriou for endless discussions and help.  To Dave Buchfuhrer, Iftah Gamzu, Arpita Ghosh, Mohammad Mahdian, George Pierrakos, Amin Saberi, Michael Schapira, Meromit Singer and Mukund Sundararajan the author wishes to express gratitude for meaningful discussions and valuable advice.


\begin{thebibliography}{10}

\bibitem{AS04}
Alexander~A. Ageev and Maxim Sviridenko.
\newblock Pipage rounding: A new method of constructing algorithms with proven
  performance guarantee.
\newblock {\em J. Comb. Optim.}, 8(3):307--328, 2004.

\bibitem{Google08}
Gagan Aggarwal, Nir Ailon, Florin Constantin, Eyal Even-Dar, Jon Feldman,
  Gereon Frahling, Monika~Rauch Henzinger, S.~Muthukrishnan, Noam Nisan, Martin
  P{\'a}l, Mark Sandler, and Anastasios Sidiropoulos.
\newblock Theory research at google.
\newblock {\em SIGACT News}, 39(2):10--28, 2008.

\bibitem{AT07}
Aaron Archer and {\'E}va Tardos.
\newblock Frugal path mechanisms.
\newblock {\em ACM Transactions on Algorithms}, 3(1), 2007.

\bibitem{ADL09}
Itai Ashlagi, Shahar Dobzinski, and Ron Lavi.
\newblock An optimal lower bound for anonymous scheduling mechanisms.
\newblock In {\em ACM Conference on Electronic Commerce}, pages 169--176, 2009.

%\bibitem{AG08}
%Yossi Azar and Iftah Gamzu.
%\newblock Truthful unification framework for packing integer programs with
%  choices.
%\newblock In {\em ICALP (1)}, pages 833--844, 2008.

%\bibitem{B07}
%Liad Blumrosen.
%\newblock Implementing the maximum of monotone algorithms.
%\newblock In {\em AAAI}, pages 30--35, 2007.

\bibitem{BN05}
Liad Blumrosen and Noam Nisan.
\newblock On the computational power of iterative auctions.
\newblock In {\em ACM Conference on Electronic Commerce}, pages 29--43, 2005.

\bibitem{BCIMS05}
Christian Borgs, Jennifer~T. Chayes, Nicole Immorlica, Mohammad Mahdian, and
  Amin Saberi.
\newblock Multi-unit auctions with budget-constrained bidders.
\newblock In {\em ACM Conference on Electronic Commerce}, pages 44--51, 2005.

\bibitem{BLM08}
Jeremy Bulow, Jonathan Levin, and Paul Milgrom.
\newblock Winning play in spectrum auctions.
\newblock {\em Working Paper}.

\bibitem{CFHK08}
Matthew Cary, Abraham~D. Flaxman, Jason~D. Hartline, and Anna~R. Karlin.
\newblock Auctions for structured procurement.
\newblock In {\em SODA}, pages 304--313, 2008.

\bibitem{DLN08}
Shahar Dobzinski, Ron Lavi, and Noam Nisan.
\newblock Multi-unit auctions with budget limits.
\newblock In {\em FOCS}, pages 260--269, 2008.

\bibitem{DS08}
Shahar Dobzinski and Mukund Sundararajan.
\newblock On characterizations of truthful mechanisms for combinatorial
  auctions and scheduling.
\newblock In {\em ACM Conference on Electronic Commerce}, pages 38--47, 2008.

\bibitem{ESS04}
Edith Elkind, Amit Sahai, and Kenneth Steiglitz.
\newblock Frugality in path auctions.
\newblock In {\em SODA}, pages 701--709, 2004.

\bibitem{FPSS02}
Joan Feigenbaum, Christos~H. Papadimitriou, Rahul Sami, and Scott Shenker.
\newblock A bgp-based mechanism for lowest-cost routing.
\newblock In {\em PODC}, pages 173--182, 2002.

\bibitem{FMPS07}
Jon Feldman, S.~Muthukrishnan, Martin P{\'a}l, and Clifford Stein.
\newblock Budget optimization in search-based advertising auctions.
\newblock In {\em ACM Conference on Electronic Commerce}, pages 40--49, 2007.

\bibitem{JM07}
Kamal Jain and Mohammad Mahdian.
\newblock Cost sharing.
\newblock In Noam Nisan, Tim Roughgarden, Eva Tardos, and Vijay~V. Vazirani,
  editors, {\em Algorithmic Game Theory}. Cambridge University Press, 2007.

\bibitem{KKT05}
Anna~R. Karlin, David Kempe, and Tami Tamir.
\newblock Beyond {VCG}: Frugality of truthful mechanisms.
\newblock In {\em FOCS}, pages 615--626, 2005.

%\bibitem{KKT03}
%David Kempe, Jon~M. Kleinberg, and {\'E}va Tardos.
%\newblock Maximizing the spread of influence through a social network.
%\newblock In {\em KDD}, pages 137--146, 2003.

\bibitem{KMN99}
Samir Khuller, Anna Moss, and Joseph~(Seffi) Naor.
\newblock The budgeted maximum coverage problem.
\newblock {\em Inf. Process. Lett.}, 70(1):39--45, 1999.

\bibitem{KG05}
Andreas Krause and Carlos Guestrin.
\newblock A note on the budgeted maximization of submodular functions.
\newblock In {\em {CMU} Technical Report}, pages CMU-- CALD-- 0 5 -- 1 0 3,
  2005.

%\bibitem{LS05}
%Ron Lavi and Chaitanya Swamy.
%\newblock Truthful and near-optimal mechanism design via linear programming.
%\newblock In {\em FOCS}, pages 595--604, 2005.

\bibitem{LLN01}
Benny Lehmann, Daniel Lehmann, and Noam Nisan.
\newblock Combinatorial auctions with decreasing marginal utilities.
\newblock In {\em ACM conference on electronic commerce}, 2001.

\bibitem{MSV08}
Vahab~S. Mirrokni, Michael Schapira, and Jan Vondr{\'a}k.
\newblock Tight information-theoretic lower bounds for welfare maximization in
  combinatorial auctions.
\newblock In {\em ACM Conference on Electronic Commerce}, pages 70--77, 2008.

%\bibitem{MR07}
%Elchanan Mossel and S{\'e}bastien Roch.
%\newblock On the submodularity of influence in social networks.
%\newblock In {\em STOC}, pages 128--134, 2007.

\bibitem{MS01}
Herv{\'e} Moulin and Scott Shenker.
\newblock Strategyproof sharing of submodular costs:budget balance versus efficiency.
\newblock {\em Economic Theory}, 18(3),  pages 511--533, 2001.

\bibitem{MN08}
Ahuva Mu'alem and Noam Nisan.
\newblock Truthful approximation mechanisms for restricted combinatorial
  auctions.
\newblock {\em Games and Economic Behavior}, 64(2):612--631, 2008.

\bibitem{M81}
R.~Myerson.
\newblock Optimal auction design.
\newblock {\em Mathematics of Operations Research}, 6(1), 1981.

\bibitem{FNW78}
G.~L. Nemhauser, L.~A. Wolsey, and M.~L. Fisher.
\newblock An analysis of approximations for maximizing submodular set functions
  ii.
\newblock {\em Math. Programming Study 8}, pages 73--87, 1978.

\bibitem{NR01}
Noam Nisan and Amir Ronen.
\newblock Algorithmic mechanism design.
\newblock {\em Games and Economic Behaviour}, 35:166 -- 196, 2001.
\newblock A preliminary version appeared in {\em STOC} 1999.

%\bibitem{R79}
%Kevin Roberts.
%\newblock The characterization of implementable choice rules.
%\newblock pages 321--349, 1979.

\bibitem{RS06}
Tim Roughgarden and Mukund Sundararajan.
\newblock New trade-offs in cost-sharing mechanisms.
\newblock In {\em STOC}, pages 79--88, 2006.

\bibitem{T03}
Kunal Talwar.
\newblock The price of truth: Frugality in truthful mechanisms.
\newblock In {\em STACS}, pages 608--619, 2003.

\bibitem{V61}
William Vickrey.
\newblock Counterspeculation, auctions, and competitive sealed tenders.
\newblock {\em The Journal of Finance}, 16(1):8--37, 1961.

\end{thebibliography}
\end{document}